\documentclass{new_tlp}
\usepackage{mathptmx}

\newcommand\bcmdtab{\noindent\bgroup\tabcolsep=0pt%
  \begin{tabular}{@{}p{10pc}@{}p{20pc}@{}}}
\newcommand\ecmdtab{\end{tabular}\egroup}

  \title[Better Paracoherent Answer Sets with Less Resources]
        {Better Paracoherent Answer Sets with Less Resources}

  \author[G. Amendola, C. Dodaro, F. Ricca]
         {GIOVANNI AMENDOLA \and CARMINE DODARO \and FRANCESCO RICCA\\
         University of Calabria, Rende, Italy\\
         \email{\{amendola,dodaro,ricca\}@mat.unical.it}}

\jdate{March 2003}
\pubyear{2003}
\pagerange{\pageref{firstpage}--\pageref{lastpage}}
\doi{S1471068401001193}

\usepackage{subfigure}
\usepackage{amssymb}
\usepackage{amsmath}
\usepackage{xspace}

\newcommand{\penalt}{\mathit{Penalty}}
\newcommand{\ssplit}{\mathit{split}}
\newcommand{\SEQ}{\mathit{SEQ}}
\newcommand{\naf}{\mathit{not} \ }
\newcommand{\la}{\leftarrow}
\newcommand{\mc}{\mathit{mc}}
\newcommand{\head}{\mathit{head}}
\newcommand{\info}{\mathit{info}}

\newtheorem{proposition}{Proposition}
\newtheorem{definition}{Definition}
\newtheorem{example}{Example}
\newtheorem{theorem}{Theorem}

\newcommand{\gcs}{\textsc{split}\xspace}
\newcommand{\gcm}{\textsc{min}\xspace}
\newcommand{\weak}{\textsc{weak}\xspace}
\newcommand{\level}{\textsc{sseq}\xspace}
\newcommand{\wasp}{\textsc{wasp}\xspace}

\newcommand{\HT}{{HT}}

\newcommand{\gap}{\mathcal{G}}

\newcommand{\sig}{\Sigma}
\newcommand{\sigk}{\sig^\kappa}

\newcommand{\cS}{{\mathcal F}}

\newcommand{\p}{P}

\newcommand{\Ik}{I^\kappa}

\newcommand{\tok}[1]{\ensuremath{{#1}^\kappa}}

\newcommand{\toht}[1]{\ensuremath{{#1}^{\HT}}}

\newcommand{\bodyn}[1]{B^-(#1)}

\newcommand{\AS}{AS}

\newcommand{\SigmaP}[1]{{\Sigma}_{#1}^{P}}
\newcommand{\PiP}[1]{{\Pi}_{#1}^{P}}

\newcommand{\MM}{\ensuremath{\mathit{MM}}}

\def\<{\langle}
\def\>{\rangle}
\def\int{\mathit{Int}}

\usepackage{tikz}
\usepackage{pgfplots}\usetikzlibrary{plotmarks}
\pgfplotsset{
	filter discard warning=false 
	, legend cell align=left
	, minor grid style={loosely dotted, lightgray}
	, major grid style={loosely dashed, lightgray}
}

\begin{document}

\label{firstpage}

\maketitle

  \begin{abstract}
Answer Set Programming (ASP) is a well-established formalism for logic programming.
Problem solving in ASP requires to write an ASP program whose answers sets correspond to solutions.
Albeit the non-existence of answer sets for some ASP programs can be considered as a modeling feature, it turns out to be a weakness in many other cases, and especially for query answering.
Paracoherent answer set semantics extend the classical semantics of ASP  to draw meaningful conclusions also from incoherent programs, with the result of increasing the range of applications of ASP.
State of the art implementations of paracoherent ASP adopt the semi-equilibrium semantics, but cannot be lifted straightforwardly to compute efficiently the (better) split semi-equilibrium semantics that discards undesirable semi-equilibrium models.
In this paper an efficient evaluation technique for computing a split semi-equilibrium model is presented.
An experiment on hard benchmarks shows that better paracoherent answer sets can be computed consuming less computational resources than existing methods.
Under consideration for acceptance in TPLP.
  \end{abstract}

  \begin{keywords}
    Answer Set Programming, Paracoherent reasoning, Semi-equilibrium models
  \end{keywords}


\section{Introduction}

In the past decades, key advances in Artificial Intelligence research were made thanks to studies in the field of Knowledge Representation and Reasoning (KRR)~\cite{DBLP:reference/fai/3}.
Among the established paradigms of KRR is Answer Set Programming (ASP)~\cite{Baral:2003:KRR:582493,DBLP:journals/cacm/BrewkaET11,DBLP:series/synthesis/2012Gebser}, which is a well-known formalism for logic programming and
non-monotonic reasoning.
%
ASP is based on the stable model (or answer set) semantics~\cite{gelf-lifs-91}, and features efficient implementations~\cite{DBLP:journals/aim/LierlerMR16,DBLP:conf/aaai/GebserMR16}, such as $\textsc{clasp}$~\cite{DBLP:conf/lpnmr/GebserKK0S15,DBLP:journals/tplp/GebserKKS19}, \wasp~\cite{AlvianoDLR15,DBLP:conf/lpnmr/AlvianoADLMR19}, and $\textsc{dlv}$~\cite{LeoneETC2006,DBLP:conf/lpnmr/AlvianoCDFLPRVZ17}.
Problem solving in ASP requires to write an ASP program whose answers sets correspond to solutions~\cite{DBLP:conf/iclp/Lifschitz99}, and then to compute these solutions resorting to an ASP solver.
The availability of efficient implementations made possible the development of concrete applications, and as a matter of fact, ASP has been successifully applied to solve complex problems in Artificial Intelligence~\cite{DBLP:conf/lpnmr/BalducciniGWN01,DBLP:journals/aim/ErdemGL16,DBLP:journals/tplp/GagglMRWW15,DBLP:conf/rr/DodaroLNR15};
Bioinformatics~\cite{DBLP:journals/jetai/CampeottoDP15};
Databases~\cite{DBLP:journals/tplp/ArenasBC03}, and
industrial applications~\cite{DBLP:journals/tplp/DodaroGLMRS16}.

The non-existence of answer sets for some ASP programs can be a modeling feature, but, as argued in~\cite{DBLP:journals/ai/AmendolaEFLM16}, it turns out to be a weakness in many other applications, such as: debugging, model building, inconsistency-tolerant query answering, diagnosis, planning and reasoning about actions.
To remedy to the non-existence of answer sets, paracoherent semantics extend the classical answer set semantics to draw meaningful conclusions also from incoherent programs.
This ASP variant has been termed \textit{paracoherent reasoning}~\cite{DBLP:conf/kr/EiterFM10}.
In particular, Eiter, Fink and Moura improved the paracoherent semantics of \textit{semi-stable models}~\cite{inou-saka-95} avoiding some anomalies with respect to basic modal logic properties by resorting to equilibrium logic~\cite{DBLP:journals/amai/Pearce06}. Thus, this paracoherent semantics is called \textit{semi-equilibrium model (SEQ) semantics}~\cite{DBLP:conf/kr/EiterFM10}.
More recently, \cite{DBLP:journals/ai/AmendolaEFLM16} noticed that, although the SEQ semantics has nice properties, it may select models that do not respect the modular structure of the program.
SEQ semantics use 3-valued interpretations where a third truth value besides \textit{true} and \textit{false} expresses that an atom is \textit{believed true}.
For instance, the incoherent logic program $P =\{ b \leftarrow \naf a;\ c\leftarrow \naf a, \naf c\}$ admits two SEQ models, say $M_1$ and $M_2$. In $M_1$, $b$ is true, $c$ is believed true, and $a$ is false; whereas in $M_2$ $a$ is believed true and both $b$ and $c$ are false.
Now, $M_1$ appears preferable to $M_2$, as, according with a layering (stratification) principle, which is widely agreed in logic programming, one should prefer $b$ rather than $a$, as there is no way to derive $a$ (note that $a$ does not appear in the head of any rule of the program).
Therefore, \cite{DBLP:journals/ai/AmendolaEFLM16}~refine SEQ-models using splitting sequences~\cite{DBLP:conf/iclp/LifschitzT94}, the major tool for modularity in modeling and evaluating answer set programs.
In particular, the refined semantics, called \textit{Split SEQ model semantics},
is able to discard model $M_2$.

The first efficient implementations of paracoherent semantics were proposed recently~\cite{DBLP:conf/aaai/AmendolaDFLR17,DBLP:conf/aaai/AmendolaD0R18}, but they only support semi-stable and semi-equilibrium semantics.
Although the Split SEQ semantics discards some undesirable SEQ models, the existing methods for computing SEQ models cannot be lifted straightforwardly to compute the refined semantics efficiently.
Consequently, \textit{no implementation of  Split SEQ semantics has been available up to now}.

In this paper, we fill this lack presenting the first efficient strategy for computing a split SEQ model.
In particular, we introduce an elegant program transformation, obtained by using weak constraints with levels~\cite{DBLP:journals/tkde/BuccafurriLR00}, that allows for computing a split semi-equilibrium model using a single call to a plain ASP solver, and prove a non-obvious correctness result.
Notably, we exploited the modularity property of split semi equilibrium models to simplify the transformation and avoid the introduction of some rules and symbols that are needed in state of the art epistemic-transformation-based methods for computing SEQ models.

We have implemented the new approach and run an experiment to validate it empirically.
Actually, no direct comparison with an alternative methods can be done, since ours is the first implementation of split SEQ models.
Nonetheless, since split SEQ models are also SEQ models, we could compare it against existing implementations for semi-equilibrium models.
As done previoulsy in the literature~\cite{DBLP:conf/aaai/AmendolaDFLR17,DBLP:conf/aaai/AmendolaD0R18}, we considered hard benchmarks from ASP competitions~~\cite{DBLP:journals/ai/CalimeriGMR16} modeling an application of paracoherent semantics to debugging~\cite{cilc19-debugging}.
The experiment demonstrates that the new method outperforms state of the art methods for computing SEQ models consuming less computational resources, i.e., it uses less memory and terminates in less time.

The paper is structured as follows:
	preliminary notions on ASP and paracoherent answer sets are reported in Section~\ref{sec:preliminaries};
	the description of strategies for\ computing split SEQ models is provided in Section \ref{sec:theor};
	the empirical validation of our approach is presented in Section~\ref{sec:experiments};
	related work is compared and discussed in Section~\ref{sec:related};
	finally, we draw the conclusion in Section~\ref{sec:conclusion}.

\section{Preliminaries}\label{sec:preliminaries}

We start with recalling answer set semantics,
and then present the paracoherent semantics of
semi-equilibrium models, and its refined version based on splitting sequences.

\subsection{Answer Set Programming}
We concentrate on logic programs over a propositional signature $\Sigma$.
A \textit{disjunctive rule} $r$ is of the form

\begin{equation}
 a_{1}\vee\cdots\vee a_{l} \leftarrow b_{1},\ldots,b_{m},\ not \ c_{1},\ldots,\ not \ c_{n},
 \label{eq:rule}
\end{equation}\vspace*{0.1cm}

\noindent where all $a_i$, $b_j$, and $c_k$ are atoms (from $\sig)$; $l > 0$, $m,n\geq 0$; $\naf$represents
\textit{negation-as-failure}. The set $H(r)=\lbrace
a_{1},...,a_{l} \rbrace$ is the \textit{head} of $r$, while $B^{+}(r)=\lbrace b_{1},...,b_{m} \rbrace$ and $B^{-}(r)=\lbrace
c_{1},\ldots,c_{n} \rbrace$ are
the \textit{positive body} and the \textit{negative body} of $r$,
respectively; the \textit{body} of $r$
is $ B(r)=B^{+}(r)\cup B^{-}(r)$. We denote by $At(r)=H(r)\cup B(r)$ the
set of all atoms occurring in $r$.
%
A rule $r$ is a \textit{fact}, if $B(r)=\emptyset$ (we then omit
$\leftarrow$);
\textit{normal}, if $| H(r)| \leq 1$;
and \textit{positive}, if $B^{-}(r)=\emptyset$.
%
A \textit{(disjunctive logic) program} $P$ is a
finite set of disjunctive rules. $P$ is called \textit{normal}
[resp.\ \textit{positive}] if each $r\in P$ is normal [resp.\ positive]. 
The set of all atoms occurring in the program $P$ is denoted by $At(P)=\bigcup_{r\in P} At(r)$.

The \textit{dependency graph} of a program $P$ is the directed graph $DG(P) = \langle V_P, E_P \rangle$ whose nodes $V_P$ are the atoms in $P$ and $E_P$ contains an edge $(a,b)$ if $a$ occurs in $H(r)$ and either $b$ occurs in $B(r)$ or in $H(r)\setminus\{a\}$.
The \textit{strongly connected components} (SCCs) of $P$, denoted $SCC(P)$, are the SCCs of $DG(P)$, which are the maximal sets of nodes $C$ such that every pair of nodes is connected by some path in $DG(P)$ with nodes only from $C$.

Any set $I\subseteq \sig$ is an \textit{interpretation};
it is a \textit{model} of a program $P$ (denoted
$I\models P$) if and only if for each rule $r\in P$, $I\cap H(r)\neq \emptyset$ if
$B^{+}(r)\subseteq I$ and $B^{-}(r)\cap I=\emptyset$ (denoted $I \models
r$).  A model $M$ of $P$ is \textit{minimal}, if and only if no model $M'\subset M$ of $P$
exists.
We denote by $\MM(P)$ the set of all minimal models of $P$ and by $AS(P)$ the set of all {\em answer sets (or stable  models)} of $P$, i.e., the set of all interpretations  $I$ such that
$I\in \MM(P^I)$, where $P^I$ is the well-known \textit{Gelfond-Lifschitz reduct}~\cite{gelf-lifs-91}
of $P$ w.r.t. $I$,  i.e., the set of rules $ a_{1}\vee ...\vee a_{l} \leftarrow b_{1},...,b_{m}$,
obtained from rules $r\in P$ of form  (\ref{eq:rule}), such that $B^-(r) \cap I= \emptyset$.
We say that a program $P$ is
\textit{coherent}, if it admits some answer set (i.e., $AS(P) \neq \emptyset$), otherwise, it is \textit{incoherent}.

In the following, we will also use \emph{constraints}, which are of the form
$$
\leftarrow b_{1},...,b_{m},not \
c_{1},...,not \ c_{n},
$$
with $m,n\geq 0$, to be considered as a shorthand for a rule
$$
\gamma \leftarrow b_{1},...,b_{m},not \ c_{1},...,not \ c_{n},not \ \gamma,
$$
using a fresh atom $\gamma$ that is not occurring elsewhere in the program.
Note that ASP solvers do normally not create any auxiliary symbols for constraints.

Moreover, we recall a useful extension of the answer set semantics by the notion of {\em weak constraint}~\cite{DBLP:journals/tkde/BuccafurriLR00}.
A weak constraint $\omega$ is of the form:
\begin{center}
	$:\sim b_1,\ldots, \ b_m, \ \naf c_{1},\ldots, \ \naf c_n. \ [w@l] $,
\end{center}
where $w$ and $l$ are nonnegative integers, representing a \textit{weight} and a \textit{level}, respectively.
Let $\Pi$ $=$ $P\cup \Omega$ , where $P$ is a set of rules and $\Omega$ is a set of weak constraints.
We call $M$ an answer set of $\Pi$ if it is an answer set of $P$.
We denote by $\Omega(l)$ the set of all weak constraints at level $l$.
For every answer set
$M$ of $\Pi$ and any $l$, the \textit{penalty} of $M$ at level $l$, denoted by $\penalt_\Pi(M,l)$, is defined as
$\sum_{\omega\in \Omega(l), \ M\models B(\omega) } w$. In case $\Pi$ is clear from the context, we omit the subscript.
For any two answer sets $M$ and $M'$ of $\Pi$,
we say $M$ is \textit{dominated} by $M'$ if
 there is $l$ s.t.
$(i)$ $\penalt_\Pi(M',l) < \penalt_\Pi(M,l)$  and $(ii)$
for all integers $k > l$, $\penalt_\Pi(M',k)$ $=$ $\penalt_\Pi(M,k)$.
An answer set of $\Pi$ is \textit{optimal} if it is not dominated
by another one of $\Pi$. We denote by $AS^O(\Pi)$ the set of all optimal answer sets of $\Pi$.



\subsection{Semi-Equilibrium Models}

Here, we introduce the paracoherent semantics of the \textit{semi-equilibrium (SEQ) models} introduced in~\cite{DBLP:conf/kr/EiterFM10}.
Consider an extended signature
$\sigk = \sig\cup\{Ka\mid a\in \sig\}$.
Intuitively, $Ka$ can be read as $a$ is believed to hold.
%
The SEQ models of a program $\p$ are obtained from its \emph{epistemic $\HT$-transformation} $\toht{\p}$, defined as follows.

\begin{definition}
\label{def:ht-trans}
Let $\p$ be a program over $\sig$. Then its epistemic $\HT$-transformation $\toht{\p}$ is obtained from $\p$ by replacing each rule $r$  of the form~(\ref{eq:rule}) in $\p$, such that $\bodyn{r}\neq\emptyset$, with:
\begin{align}
\lambda_{r,1} \vee \ldots \vee \lambda_{r,l} \vee Kc_{1} \vee \ldots \vee Kc_n & \la  b_1,\ldots, b_m, \label{eq:ep-1}\\
a_i & \la \lambda_{r,i}, \label{eq:ep-2}\\
    & \la \lambda_{r,i}, c_j, \label{eq:ep-3}\\
\lambda_{r,i} & \la  a_i, \lambda_{r,k}, \label{eq:ep-4}
\end{align}
for $1\leq i,k\leq l$ and $1\leq j\leq n$, where the $\lambda_{r,i}$, $\lambda_{r,k}$ are fresh atoms; and by adding the following set of rules:
\begin{align}
Ka & \la a,\\
Ka_1 \vee ... \vee Ka_l \vee Kc_{1} \vee ... \vee Kc_n
&\la  Kb_1,..., Kb_m,
\end{align}
\noindent for $a \in \sig$, respectively for every rule $r\in\p$ of the form~(\ref{eq:rule}).
\end{definition}
\noindent Note that for any program $\p$, its epistemic $\HT$-transformation $\toht{\p}$ is positive.
%
For every interpretation $\Ik$ over $\sig'\supseteq\sigk$, let $\gap(\Ik)=\{ Ka\in\Ik\ \mid a\not\in\Ik\}$ denote the atoms believed true but not assigned true, also referred to as the gap of $\Ik$.
Given a set $\cS$ of interpretations over $\sig'$, an interpretation $\Ik\in \cS$ is \emph{maximal canonical in $\cS$}, if no $\tok{J}\in \cS$ exists such that
$\gap(\Ik)\supset\gap(\tok{J})$.
By $\mc(\cS)$ we denote the set of maximal canonical interpretations in $\cS$.
SEQ models are then defined as \emph{maximal canonical} interpretations among the answer sets of $\toht{\p}$.

\begin{definition}
\label{def:seq}
Let $\p$ be a program over $\sig$, and let $\Ik$ be an interpretation
over $\sigk$.
Then, $\Ik\in\SEQ(\p)$
if, and only if, $\Ik\in\{ M\cap\sigk \mid M\in\mc(\AS(\toht{\p}))\}$,
where $\SEQ(\p)$ is the set of semi-equilibrium models of $\p$.
\end{definition}


\subsection{Split SEQ Models}

A set $S\subseteq At(P)$ is a {\em splitting set} of $P$, if for every rule $r$ in $P$ such that $\head(r)\cap S\neq \emptyset$ we have that
$At(r)\subseteq S$. We denote by $b_{S}(P)= \{ r\in P \mid At(r)\subseteq
S\}$
the  {\em bottom} part of $P$,
and by $t_{S}(P)=P\setminus b_{S}(P)$ the {\em top}
part of $P$ relative to $S$.
A {\em splitting sequence}
$S=(S_{1},\ldots,$ $S_{n})$ of $P$ is a sequence of splitting sets $S_i$ of $P$ such that
$S_{i}\subseteq S_{j}$ for each $i<j$.
Let $SCC(P)$ be the set of all strongly connected components of $P$, and let $(C_1,\ldots,C_n)$ be a topological ordering of $SCC(P)$. It is known that $\Gamma=(\Gamma_1,\ldots,\Gamma_n)$, where $\Gamma_j=C_1\cup\ldots\cup C_j$ for $j=1,\ldots,n$, is a splitting sequence of $P$.
So that, we obtain a \textit{stratification} for $P$ in subprograms $(P_1,\ldots,P_n)$ such that
$P_1=b_{\Gamma_1}(P)$,
and
$P_j=b_{\Gamma_{j}}(P)\setminus P_{j-1}$, for $j=2,\ldots,n$.
Given an interpretation $M_i$ over $C_i$, we denote by
$\info(M_i)$ the set of rules $\{ a \mid a\in M_i\} \cup \{ \leftarrow not \ a \mid Ka\in M_i\} \cup \{ \leftarrow a\mid a\in C_i\setminus M_i\}$.

\begin{definition}\label{def:split}
Given a topological ordering $(C_{1},\ldots,C_{n})$ of $SCC(P)$, an interpretation $M$ over $At(P)$ is a {\em semi-equilibrium model of $P$ relative to $\Gamma$} if
%
there is a sequence of interpretations $M_1,\ldots,M_n$ over $\Gamma_1,\ldots\Gamma_n$, respectively, such that
$(1)$ $M=M_n$;
$(2)$ $M_1\in \SEQ(P_1)$;
$(3)$ $M_j\in \SEQ(P_j \cup \info(M_{j-1}))$, for $j=2,\ldots,n$;
and $(4)$ $M$ is maximal canonical among the interpretations over $At(P)$
satisfying conditions $(1)$, $(2)$ and $(3)$.
The set of all semi-equilibrium models of $P$ relative to $\Gamma$ is denoted by $SEQ^{\Gamma} (P)$.
\end{definition}

Since $SEQ^{\Gamma} (P)$  is independent by the given topological ordering of $SCC(P)$~(see, Theorem 5 in~\cite{DBLP:journals/ai/AmendolaEFLM16}),
the $SCC$-models of $P$ have been defined as the set $M^{SCC}(P)$ $=$ $SEQ^{\Gamma} (P)$ for an arbitrary topological ordering of $SCC(P)$.
We will refer to them as split semi-equilibrium models.
Finally, note that $M^{SCC}(P)\subseteq SEQ(P)$.


\begin{example}
\label{es246-te-ctd}
Consider the program

\begin{center}
$P = \left\{
\begin{array}{rcl}
b &\leftarrow& \naf a;\\
d &\leftarrow& b,\ \naf c;\\
c &\la& d\
\end{array}
\right\}.$
\end{center}

Then, $(\{a\},\{b\},\{c,d\})$ is a topological ordering of $SCC(P)$, so that
$\Gamma=(\{a\},$ $\{a,b\},$ $\{a,b,c,$ $d\})$
is a splitting sequence for $P$.
Hence, $\SEQ^\Gamma(P) = \{ \{b,Kb,Kc\} \}$.
Indeed $P_1= b_{\Gamma_1}(P)=\emptyset$ and thus
$\SEQ(P_1)=\{ \emptyset \}$.
Then,
$P_2\cup \info(\emptyset) =\{ b \la \naf a,\; \la a \}$ and thus $\SEQ(P_2\cup \info(\emptyset))=\{\{b\} \}$.
Finally,
$P_3 \cup \info(\{b\}) = \{ d\leftarrow b, \naf c;\ c\la d; \ b;\ 	 \leftarrow a\} $ and thus $\SEQ(P_3\cup \info(\{b\})) = \{ \{b,Kb,Kc \}\}$.
\end{example}

In the following, we will refer to both SEQ models and split SEQ models as \textit{paracoherent answer sets}.

\subsection{Complexity Considerations}
The complexity of various reasoning tasks with paracoherent answer
sets has been analyzed in \cite{DBLP:journals/ai/AmendolaEFLM16}.
In the general case, checking whether an atom $a$ is true in \textit{some} paracoherent answer set (brave reasoning) is $\SigmaP{3}$-complete; whereas checking whether an atom $a$ is true in \textit{all} paracoherent answer sets (cautious reasoning) is $\PiP{3}$-complete.
However, computing a paracoherent answer set is feasible in F$\Delta^P_{3}$, because it is sufficient to find a paracoherent answer set that is cardinality minimal with respect to the gap.


\section{On the Computation of Split SEQ Models}\label{sec:theor}

We start to note that an efficient computation of a split semi-equilibrium model requires a deep theoretical understanding of the paracoherent semantics.
Indeed, a naive implementation of the split semi-equilibrium semantics
considers each possible path that can be generated through the splitting sequence.
Since each subprogram could have more than one answer set, one should explore an exponential number of paths.
Note that, each path generated through the splitting sequence leads to obtain a paracoherent answer set
of the last program (i.e., $P_n \cup \mathit{info}(M_{n-1})$), 
which is not necessarily a paracoherent answer set of the entire program $P$ because it must also be gap-minimal.



\begin{example}\label{ex:FinalMinimize}
Consider the program 

\begin{center}
$P = \left\{
\begin{array}{rcl}
a &\leftarrow & not \ b;\\
b &\leftarrow & not \ a;\\
c &\leftarrow & a, \ not \ c
\end{array}
\right\}.$
\end{center}

In the first layer of $P$, we have the subprogram $P_1$ $=$ $\{a \leftarrow not \ b; \ b\leftarrow not \ a\}$ whose (paracoherent) answer sets are $\{a,Ka\}$ and $\{b,Kb\}$. So that,
considering $\mathit{info}(\{a,Ka\}) \cup \{ c\leftarrow a, \ not \ c\}$, we obtain the paracoherent answer set $\{a,Ka,Kc\}$, while considering $\mathit{info}(\{b,Kb\}) \cup \{c\leftarrow a, \ not \ c\}$, we obtain the (paracoherent) answer set $\{b,Kb\}$.
Hence, $\{a,Ka,Kc\}$ cannot be a paracoherent answer set of $P$, because it has a larger gap w.r.t. $\{b,Kb\}$.
Indeed, $\gap(\{a,Ka,Kc\}) = \{ Kc \} \supset \gap(\{b,Kb\}) = \emptyset$.
\end{example}

From a computational view point, a naive approach is very expensive, as one has to enumerate all possible paracoherent models obtainable from each path to make feasible a final phase of gap minimization.

Hence, we developed a clever strategy to compute a split SEQ model.
Given a program $P$ and a topological ordering $( C_1 ,\ldots, C_n )$ of $SCC(P)$,
we construct a new program $\mathit{split}(P)$ which is the union of $\toht{\p}$ with the program $P_\gamma =\{ \gamma a \leftarrow Ka,not \ a \ |\ a\in At(P)\}$, and the
following set $\Omega$ of weak constraints.
For each $i=1,\ldots, n$, and for each atom $a\in C_i$, the weak constraint $ :\sim \gamma a \ [1:n-i] $ belongs to $\Omega$. Then, we define

$$\ssplit(P)= \toht{\p} \cup P_\gamma \bigcup_{i=1}^n\{ :\sim \gamma a \ [1:n-i] \mid a\in C_i\}.$$

We will show that an optimal answer set of $\ssplit(P)$ is a split SEQ model of $P$. Hence, in particular, it is also a SEQ model of $P$.
First, we highlight a fundamental relation between the penalty of a model at a fixed level and the set of gap atoms of that model belonging to a strongly connected component.

\begin{proposition}\label{prop:gap-penalty}
Let $M$ be an optimal answer set of $\mathit{split}(P)$. Then, for  $l=0,...,n-1$,
$$\penalt(M,l) = |\gap(M)\cap C_{n-l}|.$$ 
\end{proposition}
\begin{proof}[Proof Sketch]
Since
$\Omega(l)$ $=$ \mbox{$\{:\sim \gamma a \ [1:l] \mid a\in C_{n-l}\}$,} we obtain that
$$\penalt(M,l) = \sum_{\omega\in \Omega(l), M\models B(\omega) } w = \sum_{:\sim \gamma a\in \Omega(l), M\models \{\gamma a\} } 1 = \sum_{a\in C_{n-l}, \gamma a\in M } 1 = |\gap(M)\cap C_{n-l}|.$$
\end{proof}

Now we can prove our main result.

\begin{theorem}\label{th:main}
Let $P$ be a program and $M$ be an optimal answer set of $\mathit{split}(P)$.
Then, $M\setminus \{\gamma a \ |\ a\in At(P)\}$ is a split SEQ model of $P$.
\end{theorem}

\begin{proof}
We prove the claim by induction on the cardinality of $|SCC(P)|$.

In case of $|SCC(P)|=1$, we have a unique strongly connected component, say $C$, of $P$. Hence, we have to consider the unique topological ordering $(C)$. It leads to have the following set of weak constraints $\Omega=\Omega(1) = \{ :\sim \gamma a \ [1:0] \mid a\in C\}$. Hence, $M$ is an answer set of $P^{HT}\cup P_\gamma$ such that a minimum number of weak constraints in $\Omega$ is violated.
This means that $M$ is cardinality minimal with respect to the gap atoms. Therefore, it is also
subset minimal with respect to the gap atoms, and so, $M'=M\setminus \{\gamma a \ |\ a\in At(P)\}$ is a
semi-equilibrium model of $P$. As for $n=1$, the split semi-equilibrium models of $P$ coincide with the semi-equilibrium models of $P$, then $M'$ is a split semi-equilibrium model of $P$.

Now, assume the claim holds in case of programs with $n-1$ strongly connnected components, and we want to prove that it also holds for programs with $n$ strongly connected components.
Let $SCC(P)=\{C_1,\ldots,C_n\}$, and let $(P_1,\ldots,P_n)$ be the corresponding stratification for $P$.
Let $M_{n-1}= M\cap(C_{n-1}\cup \{Ka\ |\ a\in C_{n-1}\})$.
By construction, $M_{n-1}\in AS^O ((P^{HT}\setminus P_n^{HT})\cup P_\gamma\cup (\Omega\setminus\Omega(n)))$.
Hence, by inductive hypothesis,
$M_{n-1}\setminus \{\gamma a \ |\ a\in At(P)\}$ is a split semi-equilibrium model of $P\setminus P_n$.
Now, we have to prove that
$M'=M\setminus \{\gamma a \ |\ a\in At(P)\}\in SEQ( P_n \cup \info(M_{n-1}))$.
Consider the program
$\Pi = P_n^{HT} \cup \info(M_{n-1})\cup \Omega(n)$.
First, $(i)$ $M\in AS(P_n^{HT}\cup \info(M_{n-1}))$, as $M\in AS(P^{HT})$.
Second, $(ii)$ $M$ violates a minimum number of weak constraints in $\Omega(n)$. Indeed, by contradiction, there exists $I$ violating a strictly less number of weak constraints in $\Omega(n)$ than $M$. So that, such a $I$ dominates $M$ with respect to $split(P)$, against the assumption that $M$ is an optimal answer set of $split(P)$.
Then, by $(i)$ and $(ii)$, it holds that $M\in AS^O(\Pi)$.
Therefore, 
$M'\in SEQ(P_n \cup \info(M_{n-1}))$.
Finally, we claim that $M'$ is maximal canonical among the interpretations over $At(P)$ satisfying conditions $(1)$, $(2)$ and $(3)$ in Definition~\ref{def:split}. Assume, by contradiction, that there is an interpretation $I$ satisfying conditions $(1)$, $(2)$ and $(3)$ in Definition~\ref{def:split} such that $\gap(I)\subset \gap(M)$.
Hence, by Proposition~\ref{prop:gap-penalty},
 there is some nonnegative integer $l$ such that
$(i)$ $\penalt(I,l) < \penalt(M,l)$  and $(ii)$
for all integers $k > l$, $\penalt(I,k)$ $=$ $\penalt(M,k)$.
Then, $M$ is \textit{dominated} by $I$.
Thus, $M$ is not an optimal answer set of $\ssplit(P)$, against the hypothesis.
\end{proof}

Note that the split semi-equilibrium model (hence the semi-equilibrium model) that such an algorithm will find, generally, is not cardinality minimal among the split semi-equilibrium models of the given program.
This is coherent with complexity results, indeed computing optimal answer sets of ASP programs with weak constraints is known to be a $F\Delta^P_3$ task for disjunctive programs~\cite{DBLP:journals/tkde/BuccafurriLR00}, and our technique can be implemented by an algorithm that runs in polynomial time (indeed, it consists of two polynomial tasks: the construction of the epistemic transformation and the computation of the SCCs).

\begin{example}
Consider, for instance, the following program
\begin{center}
$P = \left\{
\begin{array}{rcl}
a &\leftarrow& not \ b; \\
b&\leftarrow& not \ a;\\
c&\leftarrow& b, \ not \ c;\\
d& \leftarrow& a, \ not \ c, \ not \ d; \\
e& \leftarrow& d\\
\end{array}\right\}$.
\end{center}
Hence, we have to consider the stratification of $P$ given by
$P_1=\{a\leftarrow not \ b;$  $b\leftarrow not \ a\}$;
$P_2=\{ c\leftarrow b, \ not \ c\}$;
$P_3=\{ d \leftarrow a, \ not \ c, \ not \ d \}$; and
$P_4=\{ e \leftarrow d \}$.
At the first step, we obtain $\{a,Ka\}$ and $\{b,Kb\}$ as SEQ models of $P_1$.
At the second step, $\{a,Ka\}$ is the SEQ model of $P_2\cup \info(\{a,Ka\})$, and $\{b,Kb,Kc\}$ is the SEQ model of $P_2\cup \info(\{b,Kb\})$.
At the third step, $\{a,Ka,Kd\}$ is the SEQ model of $P_3\cup \info(\{a,Ka\})$, and $\{b,Kb,Kc\}$ is the SEQ model of $P_3\cup \info(\{b,Kb,Kc\})$.
At the fourth and final step, $\{a,Ka,Kd,Ke\}$ is the SEQ model of $P_4\cup \info(\{a,Ka,Kd\})$, and $\{b,Kb,Kc\}$ is the SEQ model of $P_4\cup \info(\{b,Kb,Kc\})$.
Therefore, $\{a,Ka,Kd,Ke\}$ and $\{b,Kb,Kc\}$ are the split SEQ models of $P$.
However, $\{a,Ka,Kd,Ke\}$ is preferred to $\{b,Kb,Kc\}$.
Indeed, $\{b,Kb,Kc\}$ violates the weak constraint $:\sim \gamma c$, that is at a lower level than $:\sim \gamma d$ and $:\sim \gamma e$, that are violated by $\{a,Ka,Kd,Ke\}$.
\end{example}

Moreover, the split SEQ semantics allows to make a fundamental simplification of symbols in the $\HT$-epistemic transformation of the program.
Indeed, given a program $P$ and a stratification for $P$ in subprograms $(P_1,\ldots,P_n)$, whenever $P_1$, ..., $P_k$, with $k<n$, are coherent programs, we have no need to compute the $\HT$-epistemic transformation of $P_1$, ..., $P_k$, but we can directly compute the answer sets of $P_1\cup\cdots\cup P_k$. So that, if $M\in AS(P_1\cup\cdots\cup P_k)$, then
$M\cup \{Ka\ |\ a\in M\}$ is a paracoherent answer set of $P_1\cup\cdots\cup P_k$.

\begin{theorem}\label{th:optimization}
Let $P$ be a program, $(P_1,\ldots,P_n)$ be a stratification for $P$, and $k$ be the maximal number so that $P_j$ is coherent, for each $j=1,...,k$.
Then,
\begin{center}
$M^{SCC}(P) = \{ I^K\cup J \ |\ I\in \AS(P_{coh}) \wedge J\in M^{SCC}(P_{inc}^I)\}$,
\end{center}
where
$I^K=I\cup\{Ka \ |\ a\in I\}$;
$P_{coh}=P_1\cup\cdots\cup P_k$;
$P_{inc}=P_{k+1}\cup\cdots\cup P_n$;
and $P_{inc}^I$ comes from $P_{inc}$ by removing each rule $r$ s.t. $B^-(r)\cap I\neq\emptyset$, and each atom in $At(P_{inc})\cap I$.
\end{theorem}

Finally, note that to check if a program is coherent, is a well-known $\Sigma^P_2$-complete problem.
Hence, in the implementation we need to consider a \textit{sufficient} condition to detect coherent programs in polynomial time.
It is known that, if a program has no cycle in the dependency graph having an odd number of negated arcs, then it is coherent.

\begin{proposition}\label{prop:opt}
Given a program $P$, detecting a cycle in the dependency graph of $P$ having an odd number of negated arcs, can be done in linear time with respect to $|E_P|$.
\end{proposition}

\begin{proof}[Proof Sketch]
Let $DG(P)$ be the dependency graph of $P$. We consider a directed graph
$G'$ such that for each positive arc in $DG(P)$, namely $(a,b)$, we introduce a fresh node, namely $ab$, and replace $(a,b)$, by two edges $\{a,ab\}$ and $\{ab,b\}$.
Hence, if $DG(P)$ contains a cycle having an odd number of negated arcs then $G'$ contains a cycle of odd length.
The claim holds by the fact the a directed graph does not contain a directed cycle of odd length if, and only if, it is bipartite when treated as an undirected graph, i.e., it can be colored with two colors. This check can be done in linear time with respect to the number of arcs in the input graph~\cite{DBLP:books/daglib/0015106}.
\end{proof}

We conclude this section observing that Theorem~\ref{th:main} and Theorem~\ref{th:optimization} hold (without modifications) also if one considers the \textit{extended externally supported program} $P^{es}$, introduced in~\cite{DBLP:conf/aaai/AmendolaD0R18} to characterize semi-equilibrium models, in place of the $\HT$-epistemic transformation.

\section{Experiments}\label{sec:experiments}
In this section we present the results of an experimental analysis conducted to analyze the performance of the new strategy for computing a split semi-equilibrium model presented in the previous section.

\subsection{Implementation}
To compute a split semi-equilibrium model we have implemented in a rewriter tool a program transformation that takes as input an ASP program $P$ and produces as output an optimized version of $\mathit{split}(P)$.
In particular, the rewriter implements the efficient epistemic transformation $P^{es}$ of ~\cite{DBLP:conf/aaai/AmendolaD0R18}, which is known to be much more efficient than the classic $HT$-epistemic transformation, and implements the Tarjan algorithm (see e.g.,~\cite{DBLP:books/daglib/0015106}) to compute a topological order of the $SCC(P)$ and build the weak constraints $\Omega$.
During the rewriting process, the components are also subject to a (modified) two-colorability check (see Proposition~\ref{prop:opt}) applied following the topological order to optimize the output as indicated in Theorem~\ref{th:optimization}.
Thus, a split semi-equilibrium model is computed by evaluating the optimized $\mathit{split}(P)$ with the ASP solver \wasp \cite{AlvianoDLR15,DBLP:journals/tplp/AlvianoD16}.

\subsection{Experimental Setting}
A practical approach for the computation of a split semi-equilibrium model has been presented for the first time in this paper.
Therefore, it is not possible to compare the performance of our implementation with alternative approaches for computing split semi-equilibrium models.
Since all split semi-equilibrium models are semi-equilibrium models,
we focus on the task of computing a semi-equilibrium model, and compare our implementation with the state of the art ones presented in~\cite{DBLP:conf/aaai/AmendolaD0R18}, namely:
\gcs, \gcm, and \weak.
Note that, the labels correspond to the algorithm names used in~\cite{DBLP:conf/aaai/AmendolaD0R18}, where \gcs is in no way related to the split semi equilibrium semantics; rather, the name \gcs was chosen to remind the behavior of the algorithm that ``splits'' the set of gap atoms of a candidate solution to perform gap minimization (for details see~\cite{DBLP:conf/aaai/AmendolaD0R18}).
In the following, our implementation is labeled \level, since its distinguishing feature is to compute a split semi equilibrium model.

In order to perform a fair comparison, in this experiment we used the same version of the \wasp solver, and the same experimental settings described in~\cite{DBLP:conf/aaai/AmendolaD0R18}.
In particular, we considered all the incoherent instances from the latest ASP Competition~\cite{DBLP:journals/jair/GebserMR17} that feature neither \textit{aggregates}, nor \textit{choice rules}, nor \textit{weak constraints}, since such features are not supported by the paracoherent semantics~\cite{DBLP:journals/ai/AmendolaEFLM16}.
This benchmark setting is of particular interest for paracoherent reasoning since it consists of debugging hard ASP programs, one of the main motivations of paracoherent ASP.
In particular, the problem to be solved is the computation of an explanation for the non-existence of answer sets.
Execution times and memory usage were limited to 1200 seconds and 8 GB, respectively.

\begin{figure*}[t!]
	\centering
	\begin{tikzpicture}[scale=.9]
	\pgfkeys{%
		/pgf/number format/set thousands separator = {}}
	\begin{axis}[
	scale only axis
	, font=\small
	, legend style = {draw=none,fill=none, at={(0.65,1.0)}, legend columns=2}
	, xlabel={Solved instances}
	, ylabel={Execution time (s)}
	, width=0.97\textwidth
	, height=0.35\textwidth
	, ymin=0, ymax=1300
	, xmin=0, xmax=110
	, ytick={0,200,400,600,800,1000,1200}
	, xtick={0,10,20,30,40,50,60,70,80,90,100,110}
	, major tick length=2pt
	]
	\addplot [mark size=3.5pt, color=black, mark=triangle] [unbounded coords=jump] table[col sep=semicolon, y index=2] {./cactus.csv};
	\addlegendentry{\gcm}

	\addplot [mark size=3.5pt, color=black, mark=o] [unbounded coords=jump] table[col sep=semicolon, y index=3] {./cactus.csv};
	\addlegendentry{\gcs}

	\addplot [mark size=3.5pt, color=black, mark=diamond] [unbounded coords=jump] table[col sep=semicolon, y index=4] {./cactus.csv};
	\addlegendentry{\weak}

	\addplot [mark size=3.5pt, color=blue, mark=diamond*] [unbounded coords=jump] table[col sep=semicolon, y index=1] {./cactus.csv};
	\addlegendentry{\level}
	\end{axis}
	\end{tikzpicture}
\caption{Comparison of the best performing algorithms.}\label{fig:cactus}
\end{figure*}
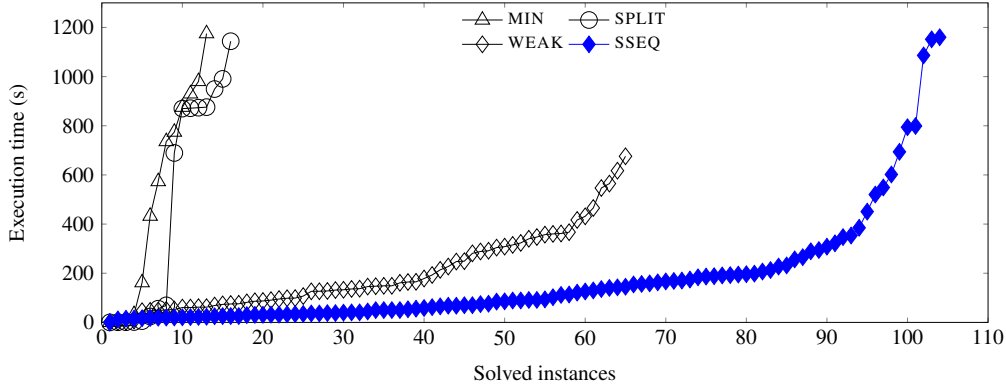

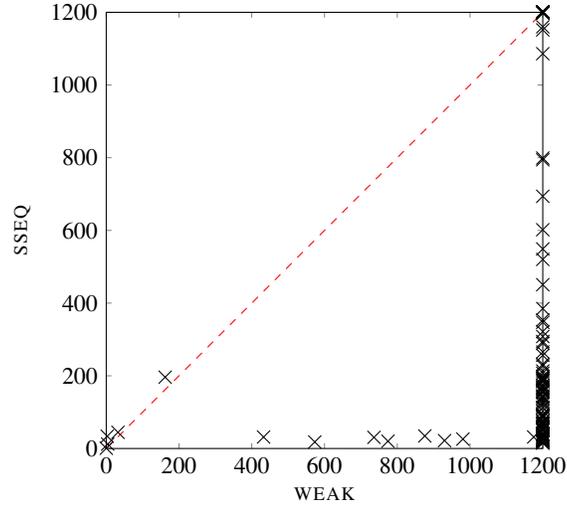
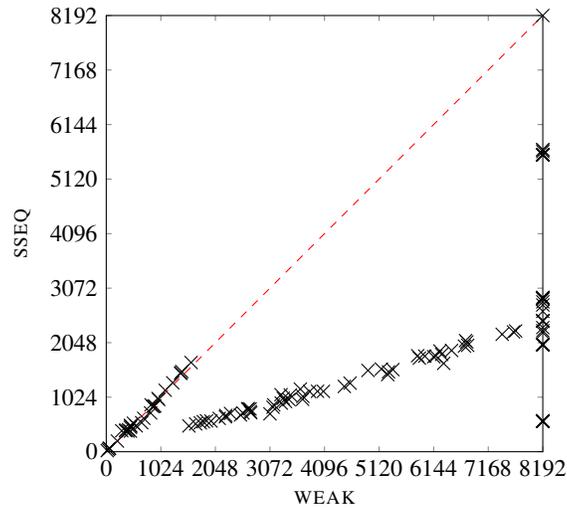
\begin{figure*}[t!]
	\centering
	\subfigure[Runtimes.]{\label{fig:scattertime}
		\begin{tikzpicture}[scale=1]
		\pgfkeys{%
			/pgf/number format/set thousands separator = {}}
		\begin{axis}[
		legend style={at={(0.03,0.93)},anchor=north west}
		, scale only axis
		, font=\small
		, x label style = {at={(axis description cs:0.5,0.02)}}
		, y label style = {at={(axis description cs:0.02,0.5)}}
		, xlabel={\weak}
		, ylabel={\level}
		, width=0.43\textwidth
		, height=0.43\textwidth
		, ymin=0, ymax=1200
		, xmin=0, xmax=1200
		, ytick={0,200,400,600,800,1000,1200}
		, xtick={0,200,400,600,800,1000,1200}
		, major tick length=2pt
		]
		\addplot [mark size=3.5pt, only marks, color=black, mark=x] [unbounded coords=jump] table[col sep=semicolon, y index=1, x index=2] {./scatter-time.csv};

		\addplot [color=red, dashed] [unbounded coords=jump] table[col sep=semicolon, x index=0, y index=0] {./scatter-time.csv};
		\end{axis}
		\end{tikzpicture}
	}\vspace*{0.2cm}
	\subfigure[Memory usage.]{\label{fig:scattermem}
	\begin{tikzpicture}[scale=1]
	\pgfkeys{%
		/pgf/number format/set thousands separator = {}}
	\begin{axis}[
	legend style={at={(0.03,0.93)},anchor=north west}
	, scale only axis
	, font=\small
	, x label style = {at={(axis description cs:0.5,0.02)}}
	, y label style = {at={(axis description cs:0.02,0.5)}}
	, xlabel={\weak}
	, ylabel={\level}
	, width=0.43\textwidth
	, height=0.43\textwidth
	, ymin=0, ymax=8192
	, xmin=0, xmax=8192
	, ytick={0,1024,2048,3072,4096,5120,6144,7168,8192}
	, xtick={0,1024,2048,3072,4096,5120,6144,7168,8192}
	, major tick length=2pt
	]
	\addplot [mark size=3.5pt, only marks, color=black, mark=x] [unbounded coords=jump] table[col sep=semicolon, y index=1, x index=2] {./scatter-mem.csv};

\addplot [color=red, dashed] [unbounded coords=jump] table[col sep=semicolon, x index=0, y index=0] {./scatter-mem.csv};
	\end{axis}
	\end{tikzpicture}
	}
\caption{Instance-wise comparison between \weak and \level.}
\end{figure*}

\subsection{Results}
The results of the experiment show that the new technique is better than state-of-the-art-approaches.
In the vast majority of considered instances the improvements are significant, as seen from the cactus plots in
Figure~\ref{fig:cactus}.
In more detail, \level solves overall 104 instances, whereas the performance achieved by \gcm, \gcs, and \weak is considerably worse, solving 13, 16, and 65 instances in the allotted time, respectively.

Figures~\ref{fig:scattertime} and \ref{fig:scattermem} show instance-by-instance comparisons for the two best-performing algorithms, i.e. \level and \weak.
We recall that each plotted point represent an instance, and a point $(x,y)$ is plotted if the two systems take $x$ and $y$ execution time (resp. memory usage) for evaluating the instance. Therefore, points below the diagonals represent instances where the system reported on the $x$-axis was slower (resp. uses more memory) than the system reported on the $y$-axis.
The graphs confirm the better performance of \level.
Indeed, in Figure~\ref{fig:scattertime}, only few instances are on the left of the diagonals, meaning that are only few instances where \level is slower than \weak.
Concerning the memory usage, Figure~\ref{fig:scattermem} clearly shows \level uses consistently less memory than \weak. We also mention that \level and \weak exceed the allotted memory limited in 1 and 83 instances, respectively. 
It is worth reporting that disabling the optimization of Theorem~\ref{th:optimization}, which is specific of split SEQ semantics, the resulting method could only solve 6 instances, and expectedly the performance w.r.t. memory consumption was also poorer, i.e., it exceeded 51 times the memory limits.

\section{Related Work}\label{sec:related}

Semantics for non-monotonic logic programs~\cite{DBLP:journals/ngc/Przymusinski91,vang-etal-91,DBLP:journals/jcss/YouY94,inou-saka-95,DBLP:journals/amai/EiterLS97,DBLP:conf/lpkr/Seipel97,Balduccini03logicprograms,pere-pint-07,DBLP:journals/japll/AlcantaraDP05,DBLP:journals/logcom/GalindoRC08,DBLP:journals/ai/AmendolaEFLM16,DBLP:journals/tplp/CostantiniF16}
that relax the definition of answer set to overcome the absence of answer sets can be considered in broader terms paracoherent semantics.
Nonetheless, the first approach to the problem of handling inconsistency in ASP programs is the semi-stable semantics by \cite{inou-saka-95}.
Later \cite{DBLP:conf/kr/EiterFM10} identified some anomalies of semi-stable semantics with respect to some epistemic properties, and proposed the semi-equilibrium semantics.
Notably, \cite{DBLP:conf/kr/EiterFM10} also introduced the term paracoherent for the semantics that provide a remedy to the absence of answer sets due to cyclic negation. 
In~\cite{DBLP:journals/ai/AmendolaEFLM16} it was demonstrated that semi-equilibrium semantics features a number of highly desirable theoretical properties for a knowledge representation language,  that are not all fulfilled by previous proposals:
(i) every consistent answer set of a program corresponds to a paracoherent answer set (\textit{answer set coverage});
(ii) if a program has some (consistent) answer set, then its paracoherent answer sets correspond to answer sets (\textit{congruence});
(iii) if a program has a classical model, then it has a paracoherent answer set (\textit{classical coherence});
(iv) a minimal set of atoms should be undefined (\textit{minimal undefinedness});
(v) every true atom must be derived from the program (\textit{justifiability}).
The first two properties ensure that the notions of answer sets and paracoherent answer sets should coincide for coherent programs; the third states that paracoherent answer set should exist whenever the programs admits a (classical) model; the last two state that the number of undefined atoms should be minimized, and every true atom should be derived from the program, respectively.
At the same time, it was observed that semi-equilibrium models do not enjoy the same nice modular composition properties of stable models (e.g., the \textit{splitting set}~\cite{DBLP:conf/iclp/LifschitzT94} modularity tool cannot be used straightforwardly).
Notably, modular composition is used in ASP for simplifying the modeling of problems (actually, the guess and check programming methodology~\cite{DBLP:conf/rweb/EiterIK09} is based on this property) and is a principle underlying the architectures of ASP systems~\cite{DBLP:journals/aim/LierlerMR16}.
The split semi-equilibrium semantics~\cite{DBLP:journals/ai/AmendolaEFLM16} solves this problem by using splitting sequences to decompose the program into hierarchically organized subprograms.
Split semi-equilibrium models are semi equilibrium models that enjoy a modularity property.

Concerning the implementation of semi-stable and semi-equilibrium semantics, we observe that they have been implemented efficiently only recently.
In particular, in~\cite{DBLP:conf/aaai/AmendolaDFLR17} a number of algorithms has been proposed, that compute paracoherent answer sets in two steps: $(i)$ an epistemic transformation of programs is applied, and $(ii)$ a strategy for computing answer sets of minimum gap is implemented by calling (possibly multiple times) an ASP solver.
The same strategy has been improved in~\cite{DBLP:conf/aaai/AmendolaD0R18} by replacing the classic epistemic transformations by more parsimonious ones (that we also adopt).
The new transformations are based on the characterization of paracoherent answer sets in terms of \textit{externally supported models}.
Neither~\cite{DBLP:conf/aaai/AmendolaDFLR17} nor \cite{DBLP:conf/aaai/AmendolaD0R18} support SSEQ semantics that is the focus of this paper.

For the sake of completeness, we mention that the algorithms used for computing paracoherent answer sets are strictly related to the computation of minimal models of propositional theories~\cite{DBLP:conf/tableaux/Niemela96,DBLP:journals/jar/BryY00,Koshimura09,DBLP:journals/ai/JanotaM16,DBLP:journals/ai/AngiulliBFP14}; the reader si referred to~\cite{DBLP:conf/aaai/AmendolaDFLR17} for a detailed discussion.

\section{Conclusion and Future Work}\label{sec:conclusion}
Paracoherent answer set semantics can draw meaningful conclusions also from incoherent programs, and in this way increase the applicability of ASP for solving AI problems~\cite{DBLP:conf/kr/EiterFM10}.
%
Practical applications are possible once efficient implementations are available, and
the complex task of computing efficiently a paracoherent answer set has been approached only recently~\cite{DBLP:conf/aaai/AmendolaDFLR17,DBLP:conf/aaai/AmendolaD0R18}.
State of the art solutions supported the semi-equilibrium semantics but cannot compute the (better) split semi-equilibrium semantics; notably, existing evaluation techniques cannot be adapted straightforwardly to accomplish this task.
We remark that, as mentioned previously, split semi-equilibrium models have to be considered better in the sense that 
they are models that respect the modular structure of the program (as observed in~\cite{DBLP:journals/ai/AmendolaEFLM16}), and, thus, they better fit the intentions of a programmer which usually exploits modularity to produce programs (e.g., by applying the guess and check methodology).

In this paper we presented a novel optimized program transformation that allows for computing a split semi-equilibrium model using a plain ASP solver.
The transformation is elegant and independent from the epistemic transformation used to define semi-equilibrium models.
Moreover, the modularity property of split semi equilibrium models allowed us to devise an optimization that further simplifies the transformed program and improves performance.
We implemented the optimized transformation and run an experiment comparing it against existing implementations for semi-equilibrium models.
Our implementation outperformed the state of the art methods in terms of both memory consumption and solving times, and it was able to solve 160\% more instances than the best alternative solution using the same ASP solver.
In conclusion, the paper shows how \textit{better semi equilibrium models can be computed also more efficiently}. 

\medskip

The availability of efficient methods for computing one paracoherent answer set makes reasonable to start approaching more complex reasoning problems connected with the enumeration of paracoherent answer sets. 
Thus, as far as future work is concerned, we plan to investigate the extension of our techniques to the implementation of cautious and brave reasoning, e.g., on the lines of~\cite{DBLP:conf/ijcai/Alviano18}. Notably, this will not be a straightforward porting.

Finally, we mention that an interesting feature work is to investigate how to extend paracoherent rewriting techniques to non-ground ASP programs.
Actually, our implementation supports non-ground ASP programs by simply disabling grounding simplifications and then using the resulting instantiation as input for the rewriting techniques applied later on.  
However, this may cause a deterioration of the performance since grounding simplifications have been shown to be useful for improving the performance of ASP solvers. Therefore, we plan to investigate if more sophisticated rewriting techniques can be directly applied to non-ground ASP programs.

\bibliographystyle{acmtrans}
\bibliography{biblio}

\label{lastpage}
\end{document}